\documentclass[sigconf,edbt]{acmart-edbt2021}

\def\BibTeX{{\rm B\kern-.05em{\sc i\kern-.025em b}\kern-.08em
    T\kern-.1667em\lower.7ex\hbox{E}\kern-.125emX}}

\usepackage{booktabs} 
\usepackage{url}
\usepackage{xcolor}
\usepackage[ruled,vlined,linesnumbered]{algorithm2e}
\usepackage{graphicx}
\usepackage{amsfonts}
\usepackage{balance}  

\usepackage{tikz}
\usepackage{neuralnetwork}
\usepackage{subfig}
\usepackage{booktabs}
\usepackage[ruled,vlined,linesnumbered]{algorithm2e}
\usepackage{enumerate}
\usepackage{diagbox}
\usepackage{rotating}
\usepackage{float}
\usepackage{hyperref}
\usepackage{pdfpages}
\usepackage{amsmath}
\usepackage{tikz}
\usepackage{comment}
\usepackage{multirow}
\usepackage{setspace}

\usepackage[inline]{enumitem}
\setcopyright{rightsretained}

\acmDOI{}

\begin{document}
\title{A Lazy Approach for Efficient Index Learning}

\settopmatter{authorsperrow=4}

\author{Guanli Liu}
\affiliation{\institution{University of Melbourne}}
\email{guanli@student.unimelb.edu.au}

\author{Lars Kulik}
\affiliation{\institution{University of Melbourne}}
\email{lkulik@unimelb.edu.au}

\author{Xingjun Ma}
\affiliation{%
  \institution{Deakin University}
}
\email{daniel.ma@deakin.edu.au}

\author{Jianzhong Qi}
\affiliation{\institution{University of Melbourne}}
\email{jianzhong.qi@unimelb.edu.au}


\begin{abstract}
Learned indices using neural networks have been shown to outperform traditional indices such as B-trees in both query time and memory.
However, learning the distribution of a large dataset can be expensive, and updating learned indices is difficult, thus hindering their usage in practical applications. In this paper, we address the efficiency and update issues of learned indices through \emph{agile model reuse}. We pre-train learned indices over a set of synthetic (rather than real) datasets and propose a novel approach to reuse these pre-trained models for a new (real) dataset. The synthetic datasets are created to cover a large range of different distributions.
Given a new dataset $\mathcal{D}_T$, we select the  learned index of a synthetic dataset  similar 
to $\mathcal{D}_T$, to index $\mathcal{D}_T$.
We show a bound over the indexing error when a pre-trained index is selected.
We further show how our techniques can handle data updates and bound the resultant indexing errors. 
Experimental results on synthetic and real datasets confirm the effectiveness and efficiency of our proposed \emph{lazy} (model reuse) approach.
\end{abstract}


\maketitle

\section{Introduction}\label{sec:intro}
\vspace{-0.5mm}
\emph{Learned indices} using neural networks have been shown to outperform traditional indices such as B-trees in both query time and memory~\cite{RMI,marcus2020benchmarking,pgm}.
Given a dataset (e.g., a database table), an index is a structure that maps the index key $p.key$ of a data point $p$ (e.g., a data record) to its storage address $p.addr$. The idea of learned indices is to train a machine learning model $\mathcal{F}$ (e.g., a neural network) to approximate the mapping from 
$p.key$ to $p.addr$. Previous work has shown that such learned indices can be simpler and more query-efficient than traditional indices. The trained model $\mathcal{F}$ can predict $p.addr$ with a bounded error range $[err_l, err_u]$, i.e.,  the data point $p$ can be found in the range of $[\mathcal{F}(p.key)+err_{l}, \mathcal{F}(p.key)+err_{u}]$ \cite{pgm}.


While learned indices have efficient query procedures, they are prone to slow building and updates, since machine learning models are expensive to train, and once trained, they are difficult to update. Even with simple models such as  linear splines, cubic splines, or linear regression, a learned index such as the \emph{recursive model index} (RMI)~\cite{RMI} is two orders  
of magnitude slower to build than a B-tree~\cite{marcus2020benchmarking}. Techniques that learn indices in a single pass such as \emph{RadixSpline}~\cite{RadixSpline} can be built faster, but they tend to produce sub-optimal indices of large sizes and lower query efficiency. The high costs in model training also prevent the retraining of learned indices for every data update. Existing learned indices~\cite{pgm,fiting_tree,ALEX} avoid model retraining by storing newly inserted points into additional structures, which inevitably adds query processing costs. This limits the applicability of learned indices in dynamic scenarios where there are frequent dataset creation or updates, which is  common in practice, for example, to index senor data or data from scientific studies (simulations)~\cite{DBLP:journals/jcheminf/ThibaultRFC14}. 

In this paper, we aim to address the efficiency issues in training and updating learned indices without hindering their query efficiency. Our solution is inspired by \emph{domain adaptation}~\cite{domain_adaptation}.
Given a model $\mathcal{M}_S$ trained on an existing (source) dataset $\mathcal{D}_S$, domain adaptation reuses $\mathcal{M}_S$ for a new (target) dataset $\mathcal{D}_T$ by fine-tuning $\mathcal{M}_S$ over $\mathcal{D}_T$. This avoids training a new model on $\mathcal{D}_T$ from scratch, which can be extremely time-consuming.

  
A key requirement for successful adaptation of $\mathcal{M}_S$ to $\mathcal{D}_T$ is that $\mathcal{D}_S$ and $\mathcal{D}_T$ should have similar distributions~\cite{DA_bounds,DA_regression}. Otherwise, $\mathcal{M}_S$ may yield large errors on $\mathcal{D}_T$. This is important in our problem as we aim to further skip fine-tuning on $\mathcal{D}_T$, to achieve fast updates. This motivates us to generate synthetic datasets to cover a wide range different data distributions 
and pre-train reusable indices on such datasets. Our dataset generation is based on the \emph{cumulative distribution function} (CDF). Given a new dataset $\mathcal{D}_T$, we measure the CDF similarity between $\mathcal{D}_T$ and the synthetic datasets. We select a model  pre-trained on a synthetic dataset similar 
to $\mathcal{D}_T$ as the index model for $\mathcal{D}_T$.

\setlength\tabcolsep{4pt} 
\begin{table}
\renewcommand\arraystretch{0.6} 
    \small  
    \centering
    \caption{Two example datasets $\mathcal{D}_S$ and $\mathcal{D}_T$.}
    \label{tab:data_sets_exp}
    \begin{tabular}{c|c|c|c|c|c|c|c|c|c|c}
        \toprule
        $\mathcal{D}_S$ & 0.1 & 0.2 & 0.3 & 0.5 & 0.6 & 0.7 & 0.8 & 0.8 & 0.9 & 1.0 \\
        \hline
        $\mathcal{D}_T$ & 0.1 & 0.2 & 0.25 & 0.3 & 0.4 & 0.5 & 0.6 & 0.8 & 0.9 & 1.0\\
        \bottomrule
    \end{tabular}
\end{table}

\begin{figure}[htp]
  \centering
  \includegraphics[width=0.3\textwidth]{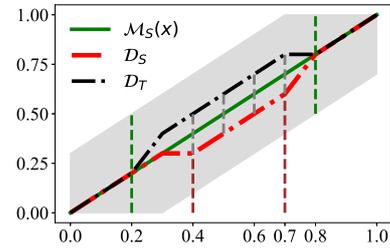}
\vspace{-1mm}
  \caption{CDFs of $\mathcal{D}_S$ and $\mathcal{D}_T$ \ref{tab:data_sets_exp} (best viewed in color).}\label{fig:cdfs_model_reuse}
  \vspace{1mm}
\end{figure}

Table~\ref{tab:data_sets_exp} and Fig.~\ref{fig:cdfs_model_reuse} illustrate two example datasets (both sorted in ascending order) and their corresponding CDFs. 
We denote the CDFs of $\mathcal{D}_S$ and $\mathcal{D}_T$ as $cdf_S(\cdot)$ and $cdf_T(\cdot)$, respectively.
In this toy example, an index model $\mathcal{M}_S$ is learned to predict the rank $p.rank$ (or percentile) of point $p \in \mathcal{D}_S$ based on its search key $p.key$, that is, $p.rank \approx \mathcal{M}_S(p.key)$ and $p.addr \approx \mathcal{M}_S(p.key)\cdot |\mathcal{D}_S|$.
This effectively learns $cdf_S(p)$. For $p \in \mathcal{D}_S$, $cdf_S(p)$ measures the probability of a value less than or equal to $p$, which is also the rank of $p$ in $\mathcal{D}_S$. 
When reusing $\mathcal{M}_S$ for $\mathcal{D}_T$, 
the additional prediction errors introduced can be bounded with respect to the dissimilarity between $cdf_S(\cdot)$ and $cdf_T(\cdot)$. This suggests that, if we can generate synthetic datasets that cover a sufficiently large area of the space of all possible CDFs, the learned indices on the synthetic datasets can be reused for any new dataset with bounded prediction errors.
Since the CDF of a dataset with $n$ points takes $O(n)$ time to compute, we further propose a histogram based approximation of CDF, with bounded 
errors, to reduce the computation time to only $O(\log n)$.

 
We use a \emph{model reuse threshold} $\epsilon \in (0, 1]$ to help determine whether to reuse a pre-trained model for a new dataset $D_T$. 
When the CDF similarity between $\mathcal{D}_T$ and a synthetic dataset $\mathcal{D}_S$ is greater than or equal to $\epsilon$, we reuse the model $\mathcal{M}_S$ pre-trained on $\mathcal{D}_S$ to index $\mathcal{D}_T$. 
Based on $\epsilon$, we further derive the maximum additional prediction error of $\mathcal{M}_S$ on $\mathcal{D}_T$, and we derive the number of synthetic datasets to be generated. 
Since our model reuse procedure is fast and flexible, we call it \emph{agile model reuse}. 
Following a similar idea, we adapt agile model reuse to handle updates. When the similarity between the CDFs of a dataset $\mathcal{D}_T$ and
its updated version $\mathcal{D}_T'$ is greater than or equal to $\epsilon$, we can reuse model $\mathcal{M}_T$ trained on $\mathcal{D}_T$ without re-training.

To showcase the applicability of our agile model reuse technique, we integrate it into the RMI learned index~\cite{RMI}. We show that agile model reuse can significantly reduce the training time of the sub-models in RMI. 
We then propose a new index structure named \emph{recursive model reuse tree} (RMRT) with built-in agile model reuse support. RMRT is designed to be adaptive to different data distributions: it builds sub-models with more layers for more dense regions of a dataset. This is particularly useful for skewed data, which has not been addressed in RMI.

In summary, our key contributions are:

\begin{enumerate}
    
    
    \item
    We propose an agile model reuse technique to pre-train a set of models on synthetic datasets and adaptively select the most suitable model to index a new dataset with respect to a model reuse threshold $\epsilon$. 
    We show how to bound the additional model prediction error given $\epsilon$.
    \item
    We propose a new index structure named RMRT, which has built-in agile model reuse support and adaptively builds an unbalanced hierarchical structure for better indexing of skewed data.
    
    \item 
    Extensive experiments on synthetic and real datasets show that agile model reuse can accelerate the building time of neural network-based learned indices by two orders of magnitude,  while retaining the lookup efficiency. Further, our agile model reuse based index RMRT is faster than RMI-based structures to build, while it outruns  all baseline models in lookup performance.
\end{enumerate}

\vspace{1mm}
\vspace{-0.1in}
\section{Related Work}\label{sec:related_work}
\vspace{-0.5mm}
A learned index~\cite{RMI,ALEX,ASLM,fiting_tree,ML-Index,pgm} learns a mapping from a search key to the storage address of a data point with a machine learning model. 
Due to limits on the learning capacity of a single model, existing learned indices such as RMI~\cite{RMI} build a hierarchy of models to index large datasets. The idea is similar to that of traditional hierarchical indices:   top-level models predict partitions of the data points (i.e., the  lower-level model in which a point is indexed), while leaf-level models  predict the storage locations. 
The training and updates of a hierarchical learned index can be very expensive, especially when neural networks are used. 
Follow-up studies aim to bound the prediction error of the learned model. For example, PGM~\cite{pgm} builds a hierarchical learned index bottom up, with a worst-case prediction error bound $\epsilon$ on every learned model. 
The building time of such learned indices is also high. 

\textbf{Update handling.} 
Updates may change the data distribution from which an index model is learned and amplify the model prediction error. Existing studies have focused on handling insertions, since points deleted can be simply flagged as ``removed'' with a light impact on  query processing. For query correctness, one may  update the prediction error bounds to ${err}_l - i$ and $err_u + i$ after $i$ insertions. 
Tighter bounds are achieved  by keeping track of the error bound drifts for a number of \emph{reference points}~\cite{DBLP:conf/sigmod/0001H19}. At query time, the closest reference points on both sides of the query point are fetched, and their error bound drifts are used to estimate the updated error bounds with a linear interpolation. 
PGM~\cite{pgm} uses two different strategies to handle insertions. For time series data insertion, it can either add a new point to the last model or add a new model to handle the new point. For arbitrary insertion, it applies the \emph{logarithmic method}~\cite{design_dynamic_data_structures} and builds a series of PGM indices for the insertions.
All these indices need extra structures to handle updates, which impact the query efficiency. 

\textbf{Domain adaptation.} 
The idea of domain adaptation is to adapt a model pre-trained on a dataset $\mathcal{D}_S$ for a new problem with a different dataset $\mathcal{D}_T$. A key step is to measure the similarity between the distributions of $\mathcal{D}_S$ and $\mathcal{D}_T$. The $L_1$ distance is a often used~\cite{DA_bounds}. It does not suit our problem because it cannot help bound the index prediction error on $\mathcal{D}_T$.  
The \emph{discrepancy}~\cite{DA_regression} is another a measure. It is designed based on testing whether the training loss differs significantly on $\mathcal{D}_S$ and $\mathcal{D}_T$.
This is inapplicable because we require a highly efficient test to determine online whether a model can be reused for $\mathcal{D}_T$. 
Typical domain adaptation techniques also fine-tune the pre-trained model on $\mathcal{D}_T$, while we skip this step for efficiency considerations.

\vspace{1mm}
\vspace{-0.1in}
\section{Agile Model Reuse}\label{sec:similarity}
\vspace{-0.5mm}

Given a new or updated dataset $\mathcal{D}_T$, we aim to construct a learned index $\mathcal{M}_T$ for $\mathcal{D}_T$ with a high efficiency. 

We first present an overview of our agile model reuse technique. We then detail its key components, including dataset similarity measurement, synthetic dataset generation, model adaptation, and error bounding. 
We will also showcase the applicability of our technique on an existing and a novel learned indices. 

    \begin{algorithm}
    \setstretch{0.8}
        \begin{small}
            \caption{Agile Model Reuse} \label{alg:mr}
            \KwIn{$\mathcal{D}_T$, $\mathcal{Q}_{MP}$}
            \KwOut{$\mathcal{M}_T$}
            \BlankLine
            \For { $<\mathcal{D}_S, \mathcal{M}_S>  \in \mathcal{Q}_{MP}$} {
                $dist \leftarrow cal\_distance(\mathcal{D}_S, \mathcal{D}_T)$\;
                
                \If{$dist \leq 1-\epsilon$} {
                    $\mathcal{M}_T \leftarrow adapt\_model(\mathcal{M}_S, \mathcal{D}_S , \mathcal{D}_T)$\;
                    return $\mathcal{M}_T$\;
                }
            }
            Train model $\mathcal{M}_T$ over $\mathcal{D}_T$\; 
            $\mathcal{M}_T.max\_abs\_err \leftarrow \mathcal{M}_T$.calc\_err($\mathcal{D}_T$)\;
            $\mathcal{Q}_{MP}.enqueue(<\mathcal{D}_T, \mathcal{M}_T>,\mathcal{M}_T.max\_abs\_err)$\;
            \Return $\mathcal{M}_T$\;
        \end{small}
    \end{algorithm}
    
\textbf{Agile model reuse overview.}
Algorithm~\ref{alg:mr} summarizes our 
agile (i.e., fast and flexible) model reuse technique.  
We pre-train models on synthetic datasets (detailed later)  which are reused to index $\mathcal{D}_T$. 
The pre-trained models are organized in a priority queue  $\mathcal{Q}_{MP}$.
Each entry in $\mathcal{Q}_{MP}$ contains the information of a synthetic dataset $\mathcal{D}_S$  and its corresponding trained model $\mathcal{M}_S$. 
The trained models are sorted by their error bounds in ascending order. Algorithm~\ref{alg:mr} traverses $\mathcal{Q}_{MP}$ (line 1), calculates the distance (dissimilarity) between $\mathcal{D}_T$ and each synthetic dataset $\mathcal{D}_S$ (line~2), and finds the first model where the distance is smaller than or equal to the model reuse threshold  $\epsilon \in (0, 1]$~(line 3). If such a model is found, the model and its error bounds are adapted based on the dataset distance (line 4, detailed later), and the adapted model is returned as $\mathcal{M}_T$ (line~5). 
Otherwise, we train a new model $\mathcal{M}_T$ for $\mathcal{D}_T$ (line 6) and obtain the error range ($err_{u}-err_{l}$, line 7). We enqueue and return the model (lines 8 and 9).

We use $\epsilon$ to control the dataset similarity in model reuse. 
A smaller $\epsilon$ allows the algorithm to return a model earlier, which may not have a high similarity with $\mathcal{D}_T$ and may lead to low prediction accuracy and high query costs.
In contrast, a larger $\epsilon$ can cost more time in  traversing $\mathcal{Q}_{MP}$ but gain a more fitted model with high prediction accuracy and low query costs. As $\epsilon$ increases in range $(0,1]$, the requirement for agile model reuse is getting higher. We elaborate on the effect of $\epsilon$ in Section~\ref{sec:experiments}.

\textbf{Dataset similarity measurement.}
A model $\mathcal{M}_S$ for dataset $\mathcal{D}_S$ 
effectively learns a CDF of $\mathcal{D}_S$. 
To reuse  $\mathcal{M}_S$ on $\mathcal{D}_T$, it is important that the CDFs of $\mathcal{D}_S$ and $\mathcal{D}_T$ are similar. We thus define the similarity between $\mathcal{D}_S$ and $\mathcal{D}_T$ by their CDFs. 



\begin{definition}[Similarity between two datasets] \label{def:similarity}
Given two datasets $\mathcal{D}_S$ and $\mathcal{D}_T$, their \emph{similarity} is defined by the maximum distance between their CDFs:
\vspace{-0.05in}
\begin{equation}\label{eq:setsim}
\small
sim(\mathcal{D}_S, \mathcal{D}_T) = 1 - \sup_x|cdf_S(x) - cdf_T (x)|
\vspace{-0.05in}
\end{equation}
Here, $\sup_x|cdf_S(x) - cdf_T (x)|$ is the maximum gap between
$cdf_S(x)$ and $cdf_T(x)$. 
We use $sim(\mathcal{D}_S,\mathcal{D}_T)$ 
and $dist(\mathcal{D}_S,\mathcal{D}_T) = 1 - sim(\mathcal{D}_S,\mathcal{D}_T)$ to denote the similarity and the distance (dissimilarity) between $\mathcal{D}_S$ and $\mathcal{D}_T$, respectively.
\end{definition}

This similarity metric is also based on the \emph{Kolmogorov–Smirnov} (KS) test~\cite{KS}, which takes $O(|\mathcal{D}_S|+|\mathcal{D}_T|)$ time to compute, assuming that both datasets are sorted already. This may be too expensive for online computation for large datasets. 
We present an approximate similarity metric for faster computation. 

Our approximate similarity metric uses \emph{relative frequency histograms} (``histograms'' for short) that discretize the data domain into bins and record relative frequencies (i.e., percentages) of the data points in each bin.
A histogram is a discrete approximation of the \emph{probability density function} (PDF) of a dataset. We use it to compute an approximation of the CDF and to compute an approximation of $dist(\mathcal{D}_S, \mathcal{D}_T)$, denoted by $dist_h(\mathcal{D}_S, \mathcal{D}_T)$.



Algorithm~\ref{alg:histogram_similarity} summarizes the computation process, 
which takes as input histograms of $\mathcal{D}_S$ and $\mathcal{D}_T$ with $m$ (a system parameter) bins each, denoted by $H_S$ and $H_T$.  We use $H_S[i]$ and $H_T[i]$ to denote the $i$-th bins and their relative frequencies. 
The sum of the probabilities of first $i$ bins of $H_S$ and $H_T$ are denoted by $P_S$ and $P_T$, i.e., $P_S = \sum_{j=0}^{i} H_S[i]$ and $P_T = \sum_{j=0}^{i} H_T[i]$. 



The algorithm computes $dist_h(\mathcal{D}_S, \mathcal{D}_T)$ ($dist_h$ for short) by looping through the bins (lines 2 to 4). 
In the $i$-th iteration ($i \in [0, m-1]$), it computes $H_S[i] + P_S$. This is the maximum $cdf_S(x)$ for any $x \in (\frac{i}{m}, \frac{i+1}{m}]$ (in our synthetic datasets, $x \in [0, 1]$), because $P_S$ has accumulated the probabilities for $x \le \frac{i}{m}$ while $H_S[i]$ further adds the probability for $x \in (\frac{i}{m}, \frac{i+1}{m}]$. Meanwhile, $P_T$ is the minimum $cdf_T(x)$ for any $x \in (\frac{i}{m}, \frac{i+1}{m}]$. Thus, $\forall x \in (\frac{i}{m}, \frac{i+1}{m}]$: 
    \vspace{-0.1in}
    \begin{equation}
    \small
    \begin{array}{cc}
        H_S[i] + P_S - P_T \ge cdf_S(x) - cdf_T(x) \\
        H_T[i] + P_T - P_S \ge cdf_T(x) - cdf_S(x)
         \end{array}
    \end{equation}
After going through all bins, we have: 
    \vspace{-0.05in}
    \begin{equation}
    \small
    \begin{array}{cc}
        dist_h(\mathcal{D}_S, \mathcal{D}_T) 
         \ge |cdf_T(x) - cdf_S(x)|, \forall x \in (0, 1]
    \end{array}
    \vspace{-0.05in}
    \end{equation}
Thus,  
        $dist_h(\mathcal{D}_S, \mathcal{D}_T) \ge dist(\mathcal{D}_S, \mathcal{D}_T)$.


\begin{algorithm}
\setstretch{0.8}
    \begin{small}
        \caption{Histogram-based-Distance} \label{alg:histogram_similarity}
        \KwIn{$H_S$, $H_T$}
        \KwOut{$dist_h$}
        \BlankLine
        $dist_h \leftarrow 0, P_S \leftarrow 0, P_T \leftarrow 0$\;
        \For { $i \in [0, m - 1]$} {
            $dist_h \leftarrow \max\{H_S[i]+P_S - P_T, H_T[i]+P_T - P_S, dist_h\}$\;
            $P_S \leftarrow P_S + H_S[i],  P_T \leftarrow P_T +  H_T[i]$\;
        }
        \Return $dist_h$\;
    \end{small}
\end{algorithm}

Using histograms to discretize CDFs reduces the  similarity computation time to $O(\log |\mathcal{D}_T| + m)$, i.e., $O(\log |\mathcal{D}_T|)$ time for $H_T$ computation and $O(m)$ time for Algorithm~\ref{alg:histogram_similarity}. Histogram $H_S$ is pre-computed since $\mathcal{D}_S$ is known. Its cost is omitted here. 

\textbf{Synthetic dataset generation.} We aim to generate a small number of 
datasets with CDFs that can be similar to those of a large number of real datasets, as bounded by threshold $\epsilon$. 

We first generate a set of CDFs to cover the space of possible CDFs. We discretize the CDF space, such that it can be covered by limited CDFs given threshold $\epsilon$. As shown in Fig.~\ref{fig:sync_cdfs}, after data normalization, all CDFs lie in a $[0, 1]^2$ space. Any CDF can be seen as a curve that starts at $(0, 0)$ and travels to $(1, 1)$ in a non-deceasing manner (in the CDF value dimension). We discrete this space with a grid, where each row has a height of $1-\epsilon$ ($\epsilon = 0.8$ in the figure), and each column has a width of $\lceil 1/(1-\epsilon) \rceil$. Consider the set $\mathcal{L}$ of polylines each starting from $(0, 0)$ and traveling to $(1, 1)$ via the grid vertices in a non-deceasing manner (in the CDF value dimension, e.g., the colored lines). Straightforwardly. given any CDF, there must be a polyline $l \in \mathcal{L}$ such that the distance  between $l$ and the CDF is bounded by $1-\epsilon$ (cf. Fig.~\ref{fig:sync_real_cdfs}).   

\vspace{-1mm}
\vspace{-0.1in}
\begin{figure}[h]
    \centering
    \subfloat[CDFs of synthetic data~\label{fig:sync_cdfs}]{
    \hspace{-3mm}
    	\includegraphics[width=0.245\textwidth]{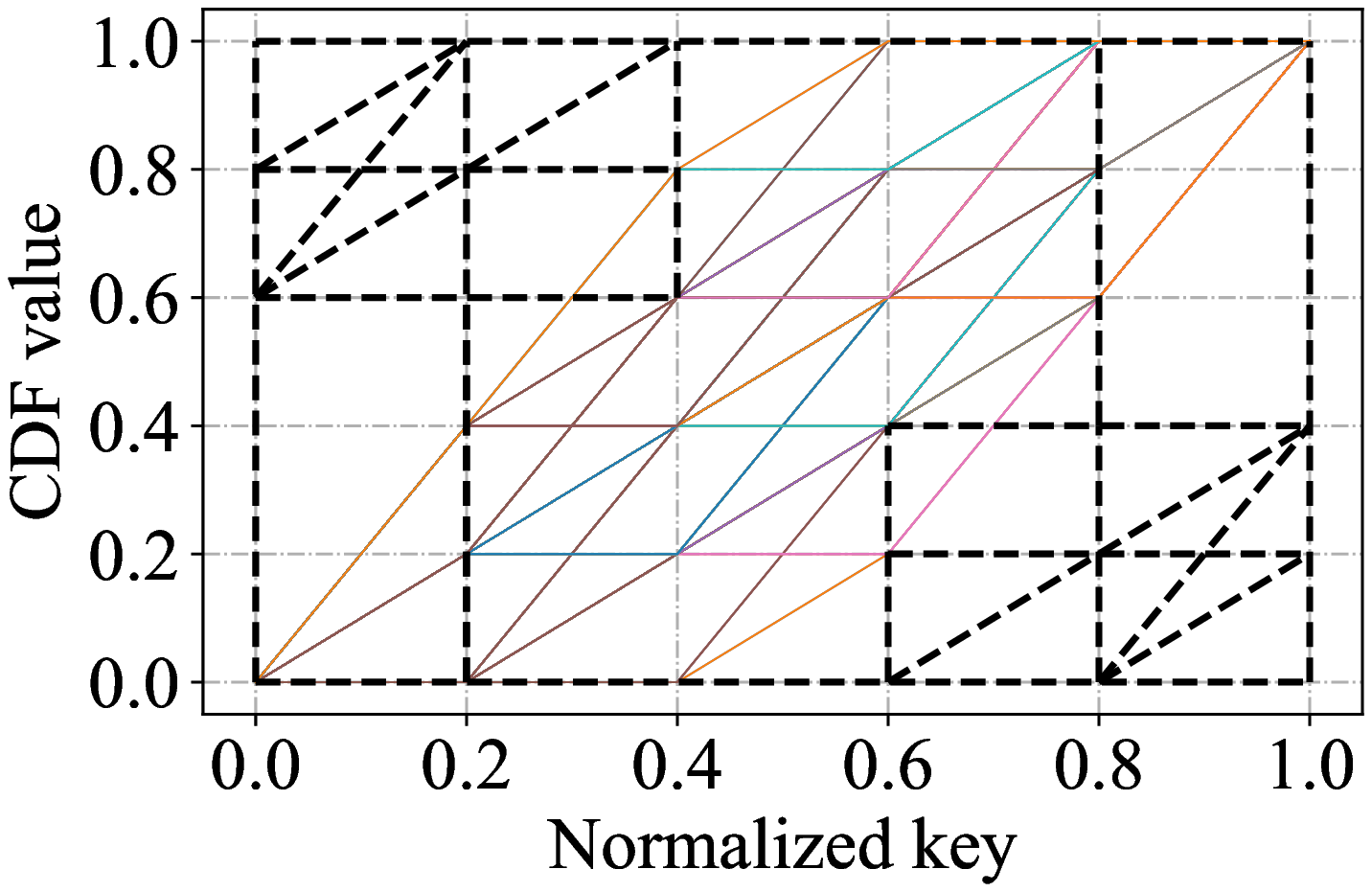}
    }
    \subfloat[CDFs of real and synthetic data~\label{fig:sync_real_cdfs}]{
    \hspace{-3mm}
    	\includegraphics[width=0.245\textwidth]{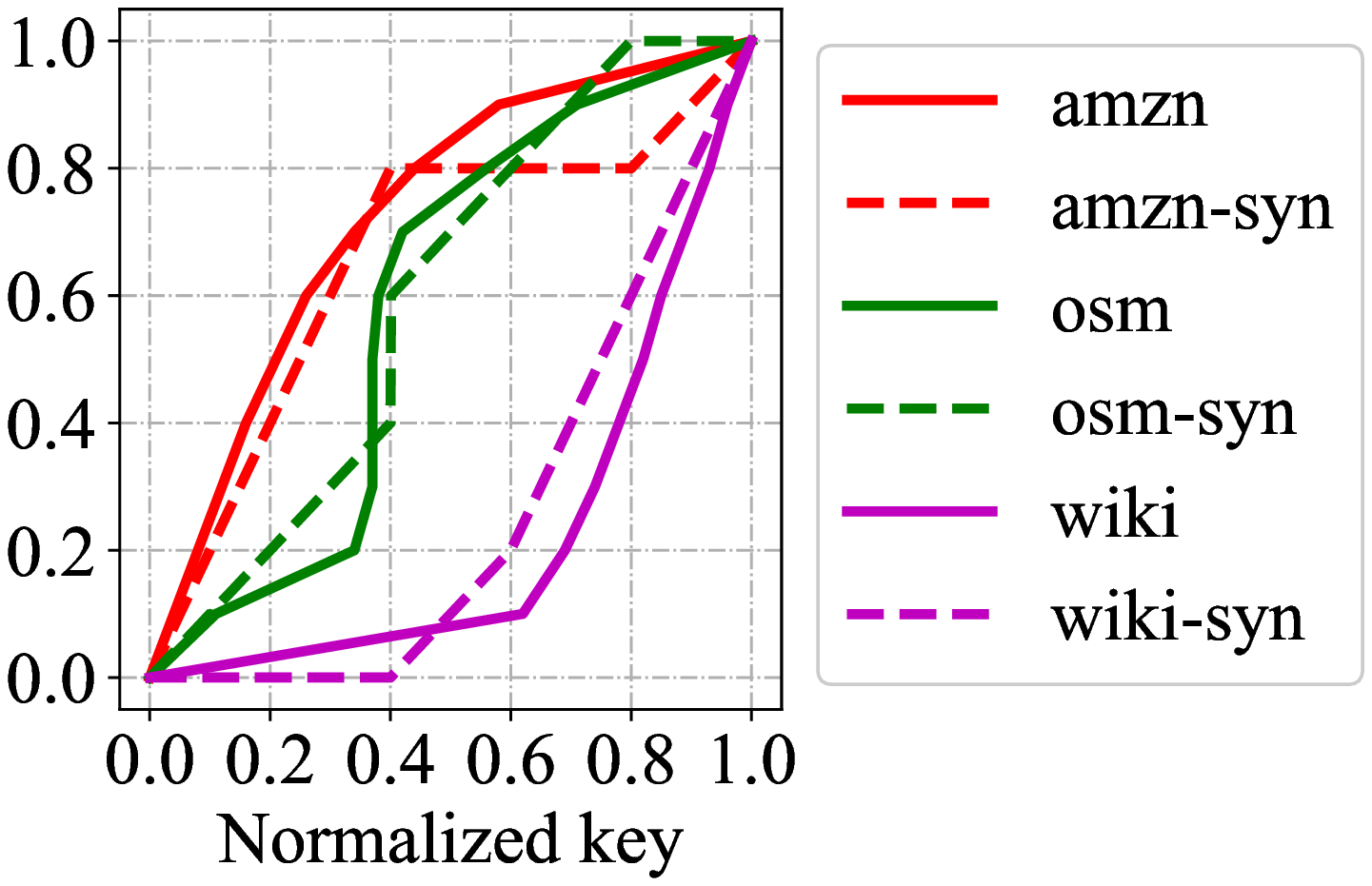}
    }
    \caption{CDF space discretization.}\label{fig:sync_data_sets}
    \vspace{1mm}
\end{figure}

The CDFs in $\mathcal{L}$ correspond to histograms with $m= \lceil 1/(1-\epsilon) \rceil$ where each bin has a probability value 
in $\{0, 1-\epsilon, 2(1-\epsilon), \ldots, 1\}$. To limit the bin value combinations and hence the number of CDFs (synthetic datasets) generated, we limit the probability value of each bin to be within $\{0, (1-\epsilon)/2,  1-\epsilon\}$, and we use $m = \lceil 2/(1-\epsilon) \rceil$ bins in the histogram heuristically. Our synthetic datasets hence will not cover the most skewed CDFs (e.g., the black polylines in Fig.~\ref{fig:sync_cdfs}). However, when a target dataset $\mathcal{D}_T$ is matched by a synthetic dataset, their CDF similarity  may be within $(1-\epsilon)/2$ rather than $1-\epsilon$, which improves the query performance. Our total number of histograms generated is: 
$\sum_{i=0}^{m} (C_{m}^{i}\cdot C_{m-i}^{\left \lfloor
(1-i(1-\epsilon))/((1-\epsilon) /2)
\right \rfloor})$, where the two combination terms represent the numbers of bins with probability values of $(1-\epsilon)$ and $(1-\epsilon)/2$, respectively. Once a histogram is generated, we generate a synthetic dataset of $ns$ key values ($ns=100$ in our experiments) based on the histogram, where the data range is $[0,1]$, and random key values are generated for each bin. 
This procedure is shown the be effective and efficient empirically.

\textbf{Model adaptation.} 
When a model $\mathcal{M}_S$ pre-trained on 
$\mathcal{D}_S$ has been selected to index $\mathcal{D}_T$, we need to adapt $\mathcal{M}_S$ based on the data domains of $\mathcal{D}_S$ and $\mathcal{D}_T$. This is because $\mathcal{M}_S$ will not work properly on a domain over which it was not trained, even if the CDFs of $\mathcal{D}_S$ and $\mathcal{D}_T$ share a similar shape. 
Let the data ranges of $\mathcal{D}_S$ and $\mathcal{D}_T$ be 
$[x_S^s, x_S^e]$ and $[x_T^s, x_T^e]$, and their data storage position ranges be $[y_S^s, y_S^e]$ and $[y_T^s, y_T^e]$, respectively. Model $\mathcal{M}_S$ has been trained to take a search key in $[x_S^s, x_S^e]$ as the input and predict an storage position in $[y_S^s, y_S^e]$. 
Here, we assume that $\mathcal{M}_S$ predicts the storage position of point $p$ directly rather than its rank (or percentile), i.e., $p.addr \approx \mathcal{M}_S(p.key)$ (instead of $\mathcal{M}_S(p.key)\cdot|\mathcal{D}_S|$ as shown in Section~\ref{sec:intro}). This simplifies the discussion but does not impact our key findings. 
To adapt $\mathcal{M}_S$ for $\mathcal{D}_T$, we take a search key in $[x_T^s, x_T^e]$, map it into $[x_S^s, x_S^e]$, and feed the mapped value into $\mathcal{M}_S$ for prediction. The predicted output needs to be mapped back into $[y_T^s, y_T^e]$ for $\mathcal{D}_T$. 

Let 
$S_{\Delta x} = \frac{x_S^e - x_S^s}{x_T^e - x_T^s}$ and $S_{\Delta y} = \frac{y_T^e - y_T^s}{y_S^e - y_S^s}$. 
The input mapping is done by a linear transformation  $\mathcal{T}_{in}(x) = a_1\cdot x + b_1$ where $a_1=S_{\Delta x}$ and $b_1=x_S^s -  x_T^s \cdot S_{\Delta x}$. This is an affine transformation that  maps the data range (i.e., $\mathcal{T}_{in}(x_T^s) = x_S^s$ and $\mathcal{T}_{in}(x_T^e) = x_S^e$) 
without changing the distribution. 
Similarly, the output mapping is done by  
$\mathcal{T}_{out}(y) = a_2\cdot y + b_2$ where $a_2=S_{\Delta y}$ and $b_2=y_T^s -  y_S^s \cdot S_{\Delta y}$. 
The mappings may incur extra costs (floating point calculation), which can be mitigated by adjusting the parameters of  $\mathcal{M}_S$. We use a linear model as an example.

\begin{lemma}\label{lemma:da_no_extra_cost}
Input and output mappings for a linear model $\mathcal{M}_S$ incur no additional prediction costs. 
\end{lemma}

\vspace{-0.1in}
\begin{proof}
Let $\mathcal{M}_S$ be a linear model $y = ax + b$, where $a$ and $b$ are parameters. The output $\tilde{y'}$ of  $\mathcal{M}_S$ is (with input mapping):  
\vspace{-0.05in}
\begin{equation}
\begin{small}
\begin{aligned}
\tilde{y'}
& = \mathcal{M}_S\big(\mathcal{T}_{in}(x)\big) = a\cdot \mathcal{T}_{in}(x) + b\\
& = a\cdot S_{\Delta x}\cdot x - a\cdot x_T^s \cdot S_{\Delta x} + a\cdot x_S^s + b
\end{aligned}
\end{small}
\end{equation}
\vspace{-0.05in}
After output mapping, the final prediction output $\tilde{y}$ is:  
\begin{equation}
\begin{small}
\begin{aligned}
\tilde{y}
& = \mathcal{T}_{out}(\tilde{y'}) = (\tilde{y'} - y_S^s)\cdot S_{\Delta y} + y_T^s = a\cdot S_{\Delta x}\cdot S_{\Delta y} \cdot x \\
& + (- a\cdot x_T^s\cdot S_{\Delta x}+ a\cdot x_S^s + b - y_S^s)\cdot S_{\Delta y} + y_T^s
\end{aligned}
\end{small}
\vspace{-0.05in}
\end{equation}
Thus, we can adapt $\mathcal{M}_S$ to a new linear model $y = a'x + b'$ for $\mathcal{D}_T$ where $a' =a\cdot S_{\Delta x}\cdot S_{\Delta y}$ and  $b' =(- a\cdot x_T^s\cdot S_{\Delta x}+ a\cdot x_S^s + b - y_S^s)\cdot S_{\Delta y} + y_T^s$. Input and output mappings can be combined with linear models without extra  prediction costs. 
\end{proof}
\vspace{-1mm}

Similar results can be derived for other models (e.g., neural models). We omit the details due to the space limit.

\textbf{Error bounding.}
Recall the prediction errors of $\mathcal{M}_S$ over $\mathcal{D}_S$, $err_{l}$ and $err_{u}$.  Given a query key $x \in \mathcal{D}_S$, the position of $x$ is bounded in $[\mathcal{M}_S(x) + err_{l}, \mathcal{M}_S(x) + err_{u}]$. 
After input and output mappings, we also need to adjust the error bounds for  $\mathcal{D}_T$. 

\vspace{-1mm}
\begin{theorem}\label{theorem:error_bounds} Let $\mathcal{M}_S$ be a  model trained on $\mathcal{D}_S$ with prediction error bounds  $err_{l}$ and $err_{u}$. 
Let $dist$ be the distance between $\mathcal{D}_S$ and $\mathcal{D}_T$.
The error bounds of $\mathcal{M}_S$ over $\mathcal{D}_T$ of size $n_T$ are:
\vspace{-0.05in}
\begin{eqnarray}
\small
err_{l}' & = & -dist\cdot n_T + err_{l}\cdot S_{\Delta y}\\
err_{u}' & = & dist\cdot n_T + err_{u}\cdot S_{\Delta y}
\vspace{-0.in}
\end{eqnarray}
\end{theorem}

\begin{proof}
For any $x \in \mathcal{D}_T$ 
and any $\mathcal{T}_{in}(x) \in \mathcal{D}_S$
, let $y$ and $y'$ be the storage positions of the corresponding data points, respectively. Then, $\mathcal{M}_S\big(\mathcal{T}_{in}(x)\big)$ is bounded by:
\vspace{-0.05in}
\begin{equation}\label{equation:min_err_3}
\small
 y' - err_{l} \geq \mathcal{M}_S\big(\mathcal{T}_{in}(x)\big)
 \vspace{-0.05in}
\end{equation}
After output mapping,  $\mathcal{T}_{out}\big(\mathcal{M}_S\big(\mathcal{T}_{in}(x)\big)\big)$ is the predicted position of $\mathcal{M}_S$ over $\mathcal{D}_T$, which 
is bounded by:
\vspace{-0.05in}
\begin{equation}\label{equation:min_err_1}
\small
y - err_{l}'\geq \mathcal{T}_{out}\Big(\mathcal{M}_S\big(\mathcal{T}_{in}(x)\big)\Big)
\end{equation}
\vspace{-0.05in}
Since mappings are monotonic, Inequality~(\ref{equation:min_err_3}) is  rewritten as:
\vspace{-0.05in}
\begin{equation}\label{equation:min_err_4}
\begin{small}
\begin{array}{l}
\mathcal{T}_{out}(y' - err_{l}) 
= \mathcal{T}_{out}(y') - \mathcal{T}_{out}(err_{l}) + b'
\geq \mathcal{T}_{out}\Big(\mathcal{M}_S\big(\mathcal{T}_{in}(x)\big)\Big)
\end{array}
\end{small}
\end{equation}
where $b'$ is the interception of $\mathcal{T}_{out}(\cdot)$.
Given $dist$ as the distance between 
$\mathcal{D}_T$ and $\mathcal{D}_S$, 
after input and output mappings, we  have 
$|\mathcal{T}_{out}(y') - y| \leq dist\cdot n_T$, i.e., $y + dist\cdot n_T \geq \mathcal{T}_{out}(y')$.
Combining with Inequality~(\ref{equation:min_err_4}), we have 
$y + dist\cdot n_T - \mathcal{T}_{out}(err_{l}) + b'\geq \mathcal{T}_{out}\Big(\mathcal{M}_S\big(\mathcal{T}_{in}(x)\big)\Big)$. 
Given this inequality, to satisfy Inequality~(\ref{equation:min_err_1}), we enforce  $y  - err_{l}' \geq y + dist\cdot n_T - \mathcal{T}_{out}(err_{l}) + b' $.
Thus, $err_{l}' \leq -dist\cdot n_T + err_{l}\cdot S_{\Delta y}$, where $ S_{\Delta y}$ is the slope of $\mathcal{T}_{out}(\cdot)$.
Letting $err_{l}'=-dist\cdot n_T + err_{l}\cdot  S_{\Delta y}$ will ensure query correctness.
Similarly, we can derive 
$err_{u}' = dist\cdot n_T + err_{u}\cdot S_{\Delta y}$.
\end{proof}

\vspace{-0.1in}
  \begin{figure}[htp]
      \centering
      \includegraphics[width=0.45\textwidth]{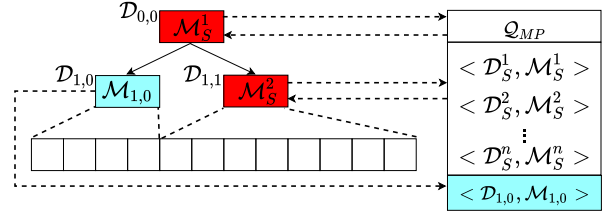}
      \caption{Building RMI with agile model reuse  
        ($\epsilon = 0.8$).
      }\label{fig:model_reuse}
      \vspace{1mm}
  \end{figure}
  
\textbf{Learned indices with agile model reuse.} 
We showcase the applicability of agile model reuse over existing learned indices by building a two-layer RMI~\cite{RMI}, as shown in Fig.~\ref{fig:model_reuse}. We first compute $sim(\mathcal{D}_S^{1}, \mathcal{D}_{0,0})$ between the full dataset $\mathcal{D}_{0,0}$ to be indexed and a synthetic dataset $\mathcal{D}_S^{1}$, which has a pre-trained model $\mathcal{M}_S^{1}$ in $\mathcal{Q}_{MP}$. Suppose  $sim(\mathcal{D}_S^{1}, \mathcal{D}_{0,0})=0.9 \ge \epsilon=0.8$. Then, $\mathcal{M}_S^{1}$ is reused over $\mathcal{D}_{0,0}$. 
For $\mathcal{D}_{1,0}$, a subset to be indexed by a child model in RMI, 
we cannot find a synthetic dataset with a similarity greater than or equal to 0.8. We train a model $\mathcal{M}_{1,0}$ over $\mathcal{D}_{1,0}$ and put it into $\mathcal{Q}_{MP}$ for reuse later. 
For the other subset $\mathcal{D}_{1,1}$, we find another synthetic dataset $\mathcal{D}_S^{2}$ with   
$sim(\mathcal{D}_S^{2},\mathcal{D}_{1,1})=0.85 \ge 0.8$. Its model $\mathcal{M}_S^{2}$ is reused over $\mathcal{D}_{1,1}$, which completes the RMI.

\textbf{Recursive Model Reuse Tree.}
In RMI, 
the number of layers and 
the number of models in each layer is fixed. 
If the data is skewed, the cardinality of the subsets assigned to different models can vary considerably, resulting in high prediction errors and search costs on some models. To address this issue, we design a learned index structure with built-in agile model reuse support named the \emph{recursive model reuse tree} (RMRT).

 Suppose that 
 the models used in RMRT have the same learning capacity (e.g., neural networks of the same structure), which can fit at most $N$ points each.
 When $|\mathcal{D}_T|$ 
is greater than $N$, we first learn a model $\mathcal{M}_{0, 0}$ to predict the points into $B$ partitions where $B$ is a system parameter. Then, recursively, for points predicted to partition $i$, we learn another model $\mathcal{M}_{1, i}$ to partition them. This process continues, until each  partition has at most $N$ points, which is indexed by a learned model. Agile model reuse is applied whenever a model is needed in this process. 
Fig.~\ref{fig:rmrt} gives an example with $N=4$ and $B=2$. Model $\mathcal{M}_{0,0}$ predicts two partitions (i.e., subsets) $\mathcal{D}_{1,0}$ and $\mathcal{D}_{1,1}$ that contain the first four and the last eight points, respectively. 
Further partitioning is needed for $\mathcal{D}_{1,1}$. A model 
$\mathcal{M}_{1,1}$ is learned for this, creating two partitions of 
size $N$ each. The partitioning then stops.

\vspace{-1mm}
\vspace{-0.1in}
\begin{figure}[h]
    \centering
    \subfloat[RMRT ($N$=4, $B$=2)~\label{fig:rmrt}]{

    	\includegraphics[width=0.22\textwidth]{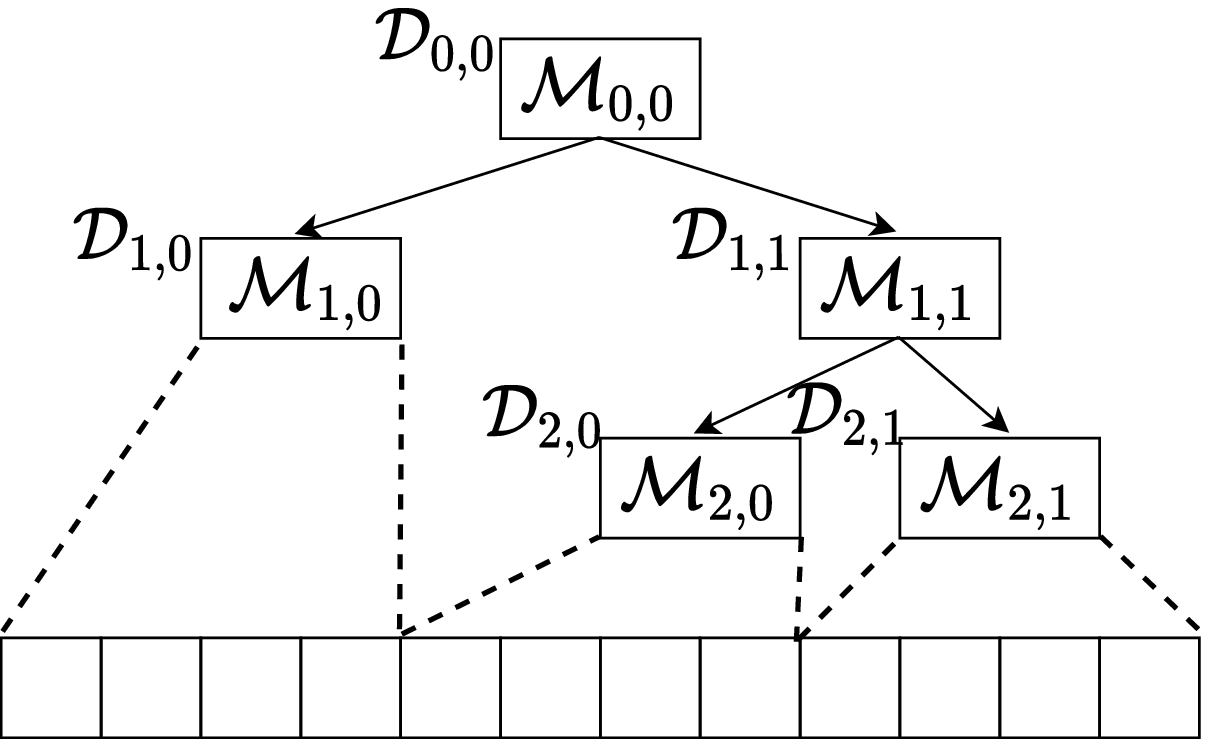}
    }
    \subfloat[Insertion handling~\label{fig:update}]{

    	\includegraphics[width=0.24\textwidth]{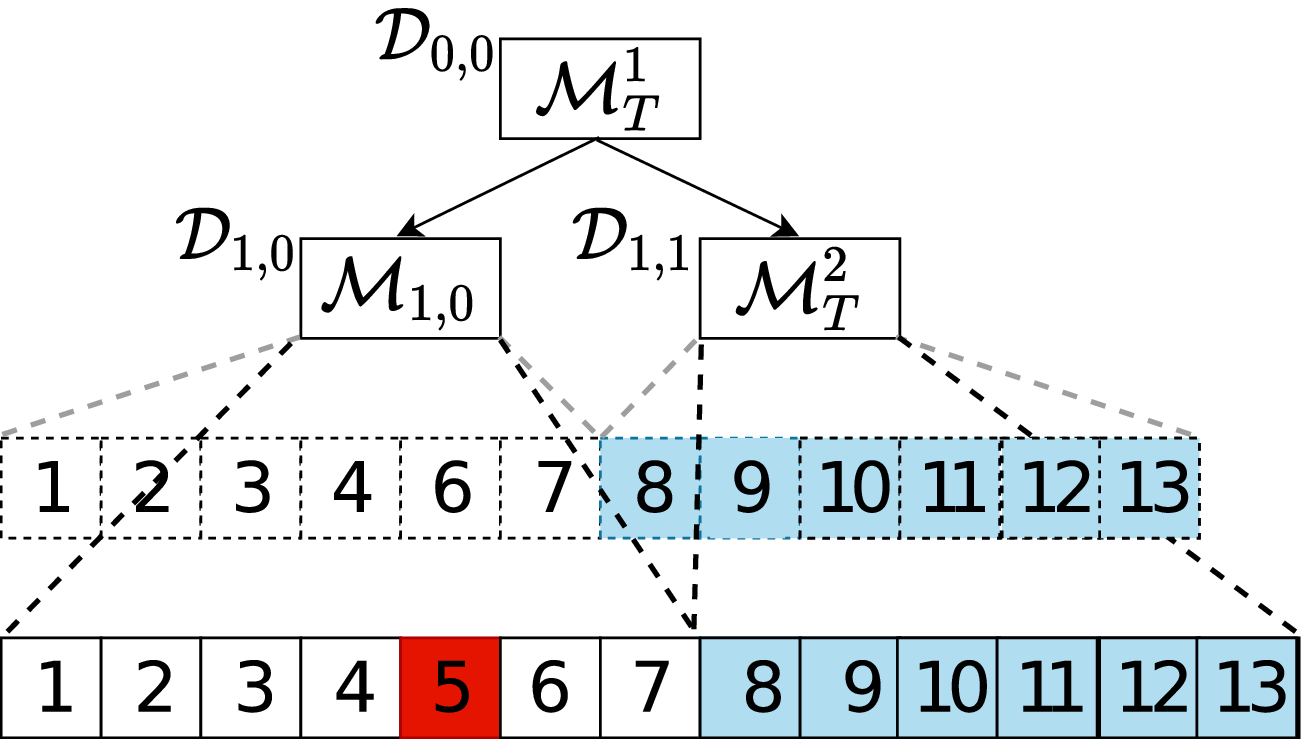}
    }
    \caption{RMRT and insertion handling examples}
    \vspace{1mm}
\end{figure}

\vspace{1mm}
\vspace{-0.1in}
\section{Update Handling}
\vspace{-0.5mm}
We focus on insertions. Deletions can be implemented simply as a point query and marking the queried point as ``deleted''.

To handle insertions, we examine their impact on the CDF. As shown in Fig.~\ref{fig:update}, when data point 5 is inserted, only the CDF of $\mathcal{D}_{1,0}$ is impacted in the second layer of the recursive model.  
For $\mathcal{D}_{1,1}$ which is not impacted, we can simply add 1 to its model error bounds.  
For $\mathcal{D}_{1,0}$, we have to check whether the reused model $\mathcal{M}_{1,0}$ can still meet the similarity bound defined by threshold $\epsilon$. To enable efficient checks, we propose a bound on the maximum number of insertions without requiring model updates. 
\begin{lemma}\label{lemma:insert_bound}
Let $\mathcal{D}$ be a dataset with cardinality $n_{\mathcal{D}}$ and $\mathcal{M}_D$ be a  model over $\mathcal{D}$.  Let $sim$ be the similarity between $\mathcal{D}$ and the  dataset from which $\mathcal{M}_D$ is trained, which can be $\mathcal{D}$ or a synthetic dataset.  
If there are less than $\frac{(sim - \epsilon)}{1+ \epsilon - sim} \cdot n_{\mathcal{D}}$ insertions on $\mathcal{D}$, we can still reuse model $\mathcal{M}_D$ for the resultant dataset $\mathcal{D'}$. 
\vspace{-0.5mm}
\end{lemma}
\begin{proof}
After $n_i$ insertions, the CDFs $cdf_D(x)$ and $cdf_{D'}(x)$ may become different. In the worst case, all new data points are inserted at the same position, where the difference between $cdf_D(x)$ and $cdf_{D'}(x)$ is bounded by $dist(\mathcal{D}, \mathcal{D}') \leq \frac{n_i}{n_i + n_{\mathcal{D}}}$. Recall that  $\mathcal{M}_D$ is reused with $sim \geq \epsilon$ or trained on $\mathcal{D}$ ($sim=1\geq \epsilon$). We can use the gap $sim- \epsilon$ as a buffer to accommodate the CDF drift caused by the insertions. 
According to the transitivity of inequalities, if $\frac{n_i}{n_i + n_{\mathcal{D}}} \leq sim - \epsilon$, there is no need to rebuild a model over $\mathcal{D}'$, since the impact on the CDF of $\mathcal{D}$ by the insertions cannot exceed the error bound $sim - \epsilon$. 
Thus, we can derive a bound on the number of insertions $n_i$ as $n_i \leq \frac{(sim - \epsilon)}{1+ \epsilon - sim} \cdot n_{\mathcal{D}}$\end{proof}

According to Lemma~\ref{lemma:insert_bound},  insertions can be handled without model rebuilt when their number does not exceed the bound. 
When a new data point is inserted, we find the target insertion position through a point query and obtain the corresponding model $\mathcal{M}$.  If the number of insertions on  $\mathcal{M}$ has not exceeded the bound, the insertion is completed. Otherwise, we only rebuild the model indexing the inserted data point.


\vspace{1mm}
\vspace{-0.1in}
\section{Experiments}\label{sec:experiments}
\vspace{-0.5mm}
All experiments are done on a 64-bit machine with a 3.60 GHz Intel i9 CPU, RTX 2080Ti GPU, 64~GB RAM, and a 1 TB hard disk. We use \emph{PyTorch}~1.4~\cite{pytorch} and its C++ APIs to implement the learned indices based on GPU. The linear and neural network models are implemented using \emph{Scikit-learn}~\cite{sklearn} and \textit{PyTorch}, respectively.

\textbf{Competitors.} We compare with both traditional and learned indices: 1) traditional \textbf{BTree}~\cite{btree} which is a C++ based in-memory B+ tree from the STX B+ Tree package; 
2) \textbf{RMI}~\cite{RMI} which is the linear RMI model from the SOSD benchmark~\cite{marcus2020benchmarking};
3) \textbf{RMI-NN} which is our implementation of the neural network RMI model;
4) \textbf{PGM}~\cite{pgm} which is a piecewise geometric model index; and 5)~\textbf{RS~}~\cite{RadixSpline} which is a single-pass learned index.


\textbf{Proposed models.} We study the performance of the following proposed and adapted models:
 1) \textbf{RMI-MR} which is the linear RMI model enhanced by agile model reuse; 2) \textbf{RMI-NN-MR} which is the neural RMI model enhanced by agile model reuse; and 3) \textbf{RMRT} which is our proposed learned index.

\textbf{Implementation details.}
For BTree, RMI, PGM, and RS, we use their published source code and default configurations. For RMI-MR, we adapt the original model training code to include agile model reuse.
For neural network based models including RMI-NN, RMI-NN-MR and RMRT, we use feedforward neural networks each with one hidden layer of four neurons.
We use $N=10^{6}$ as the RMRT model size threshold, which shows strong empirical performance. We set the default value for $\epsilon$ to 0.9.

We summarize the number of synthetic datasets (each with 100 points) and the time to pre-train models on them in
Table~\ref{tab:synthetic_gen_pretrain_epsilon}. 
Note that we use $m=12 < \lceil 2/(1-\epsilon) \rceil$ when $\epsilon=0.9$. As a result, the number of datasets generated is bounded in 10,000; all pre-trained models can be loaded in memory
within a second (30 MB in size); and the total  model comparison time to build an index  in any of the experiments is also within a second. 

\vspace{1mm}
 \begin{table}[h]
 \renewcommand\arraystretch{0.5} 
    \centering
        \small
        \caption{Summary of Synthetic Datasets.}
        \label{tab:synthetic_gen_pretrain_epsilon}
    \begin{tabular}{l|l|l|l|l|l}
    \toprule   
        $\epsilon$ &  0.5 & 0.6 & 0.7 & 0.8 & 0.9 \\
         \midrule
         \midrule
        Number of bins ($m$) & 4 & 5 & 7 & 10 & 12\\
        \midrule
        Number of datasets  & 19 & 95 & 987 & 8,953 & 1,221\\
         \midrule
       Total model training time (s) &  2.1 & 8.8 & 63.5 & 839.5 & 109.1\\
        \bottomrule
    \end{tabular}
\end{table}

\textbf{Datasets.} 
Following SOSD~\cite{marcus2020benchmarking}, we use four real datasets:
    \textbf{amzn} (default) -- an Amazon book popularity dataset,
    \textbf{face} -- a Facebook user ID dataset, 
    \textbf{osm} -- an OpenStreetMap cell ID dataset, and 
    \textbf{wiki} -- a Wikipedia edit timestamp dataset.
    We further generate skew datasets 
    from uniform data by raising a key value $x$ to its powers $x^\alpha$ ($\alpha = 1,3,5,7,9$), 
    following~\cite{DBLP:journals/pvldb/QiTCZ18}.
    Each dataset contains 200 million unsigned 64-bit integer keys.

\textbf{Performance metrics.}
We report the index build time, lookup (i.e., point query) time, and update (insertion) time.

\vspace{1mm}
\vspace{-0.1in}
\subsection{Results}
\vspace{-0.5mm}

\textbf{Index build time on real datasets.} 
The four real datasets have different distributions. 
They lead to different index build times as shown in Fig.~\ref{fig:build_time_distribution}. 
BTree is the fastest to build and is little impacted by the data distribution, due to its simple building procedure.
With agile model reuse, RMI-NN-MR is two orders of magnitude faster than RMI-NN on all four datasets, while RMI-MR is also consistently faster than its non-model-reusing counterpart RMI. Our   RMRT further outperforms RMI-NN-MR and PGM, and the advantage is up to 74\% (3.0s vs. 11.7s on face) and 32\% (5.0s vs. 7.4s on face) than PGM, respectively. RS is faster than RMRT, due to its single-pass procedure, but its lookup time is substantially higher than that of RMRT, which is detailed next.

\vspace{-0.1in}
\begin{figure}[h]
    \centering
    
    \subfloat[Index build time~\label{fig:build_time_distribution}]{
    	\includegraphics[width=0.255\textwidth]{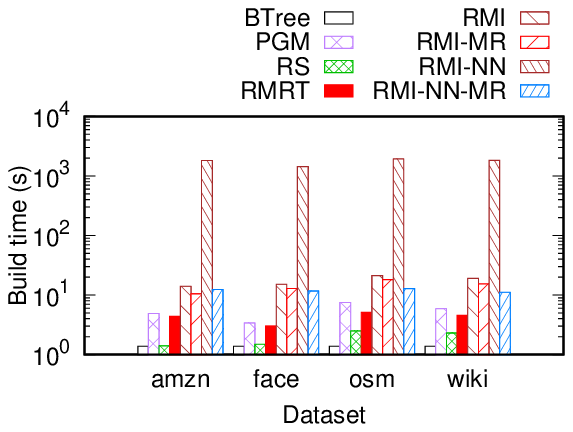}
    }
    \subfloat[Lookup time~\label{fig:query_time_distribution}]{
            \hspace{-6mm}
    	\includegraphics[width=0.255\textwidth]{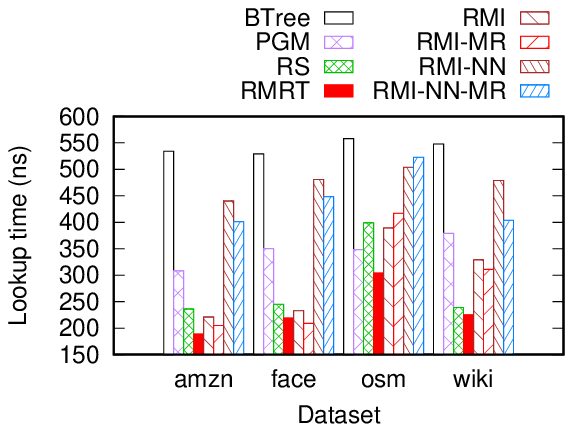}
    }
    \caption{Build and lookup time on real  datasets.}\label{fig:build_and_query_time}
        \vspace{0.5mm}
\end{figure}

\textbf{Lookup time on real datasets.} Each real dataset has 10 million random lookup keys, and we report their average lookup time in Fig.~\ref{fig:query_time_distribution}.  
RMRT is the fastest over all four datasets. On amzn, RMRT (189 ns) is 46\%, 35\%, 14\%, 7\%, 52\%, 38\%, and 20\%
faster than BTree (534 ns), RMI-NN (440 ns),  RMI (221 ns), RMI-MR (205 ns), RMI-NN-MR (401 ns),  PGM (308 ns),  and RS (236 ns), respectively. 
RMI-MR and RMI-NN-MR have better performance than RMI and RMI-NN over amzn, face, and wiki, since these three datasets are well distributed and can be  fitted by the pre-trained models. 
For osm, 
it differs more from the synthetic datasets, 
where  
the reused model has increased lookup costs. 
We further note that RS index stores spline points, the number of which are a decisive factor in RS  lookup time. To obtain the lookup performance shown here, the index size of RS is about an order of magnitude larger than that of RMRT.

\vspace{-0.1in}
\begin{figure}[h]
    \centering

    \subfloat[Index build time~\label{fig:build_time_skew}]{
    	\includegraphics[width=0.255\textwidth]{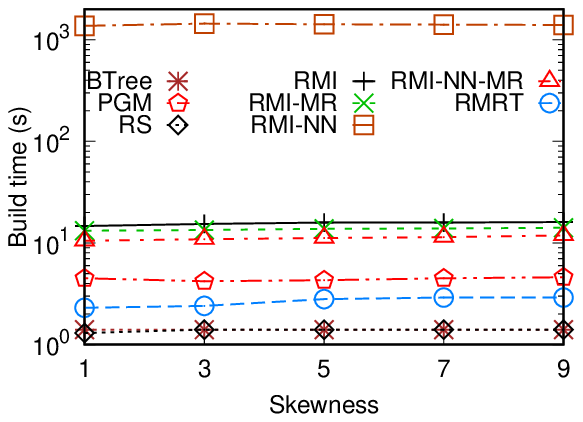}
    }
    \subfloat[Lookup time~\label{fig:query_time_skew}]{
            \hspace{-6mm}
    	\includegraphics[width=0.255\textwidth]{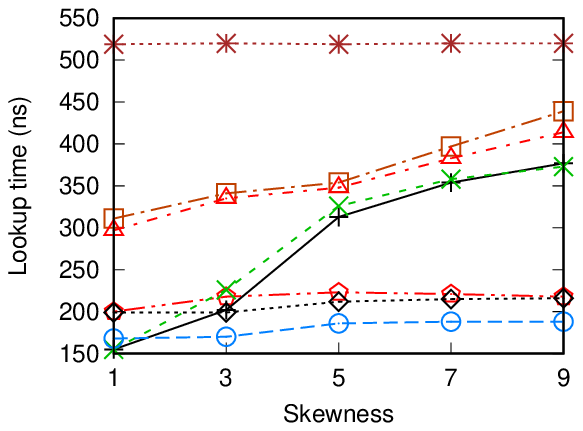}
    }
    \caption{Build and lookup time on skew datasets.}\label{fig:build_and_query_skew_time}
    \vspace{0.5mm}
\end{figure}

\textbf{Index build time on skewed datasets.} 
Since our RMRT targets skewed data, we further test the indices on synthetic data with increasing skewness. As Fig.~\ref{fig:build_time_skew} shows, the index build times are less impacted by data skewness. 
This is consistent with the results on real datasets, where the index build times are also  similar across different datasets. RMI-NN-MR and RMI again outperform RMI-NN and RMI, respectively, while our RMRT is only slightly slower than BTree and the single-pass learned index RS.

\textbf{Lookup time on skewed datasets.} 
As Fig.~\ref{fig:query_time_skew} shows, 
our RMRT again yields the best lookup performance on all  skewed datasets (skewness = 1 denotes uniform data), and its performance is stable as data skewness increases, confirming its capability to adapt to skewed data. In contrast, the four RMI-based indices have fast increasing lookup times when the data skewness increases, as analyzed in Section~\ref{sec:similarity}.  
BTree, PGM, and RS are also less impacted because their designs are based on worst-case scenarios.

\textbf{Index build time under varying $\epsilon$.} Table~\ref{tab:epsilon_update_time} shows that, as $\epsilon$ increases, the build times of both RMI-NN-MR and RMRT increase.   
This is because a larger $\epsilon$ requires a higher similarity for model reuse, thus more  datasets are examined.
For our RMRT, the build time decreases initially. This is because a small $\epsilon$ ($\epsilon=0.5$) cannot fit the datasets well, which creates uneven partitions that take more models to fit. As $\epsilon$ increases beyond $\epsilon=0.6$, the build time of RMRT rises again. Note that RMI-MR shows a similar trend to RMI-NN-MR and is omitted for conciseness.




\textbf{Lookup time under varying $\epsilon$.} Table~\ref{tab:epsilon_update_time} also shows that,
as $\epsilon$ increases, the lookup times of both models decrease. 
This is because better-fitted models are selected for larger $\epsilon$ which bring shorter search ranges. We see that the benefit in lookup outweighs the extra index building costs when using a larger $\epsilon$. 

\vspace{1mm}
\setlength\tabcolsep{4pt}
\begin{table}[h]
\renewcommand\arraystretch{0.5} 
    \centering
        \small
        \caption{Build, lookup and insertion time under varying $\epsilon$.}
        \label{tab:epsilon_update_time}
    \begin{tabular}{l|l|l|l|l|l|l}
    \toprule   
    ~ & $\epsilon$ & 0.5 & 0.6 & 0.7 & 0.8 & 0.9 \\
                      \midrule
                      \midrule
    \multirow{2}{*}{Build (s)} & RMI-NN-MR & 7.6 & 7.9 & 8.3 & 11.8 & 14.4\\
                     ~ & RMRT  & 4.2 & 3.7 & 4.1 & 4.3 & 4.4\\
                      \midrule
    \multirow{2}{*}{Lookup (ns)} & RMI-NN-MR & 572 & 570 & 470 & 466 & 401 \\
                     ~ & RMRT  & 349 & 321 & 274 & 223 & 189 \\
                    \midrule
    \multirow{2}{*}{Insertion (ns)} & RMI-NN-MR & 99 & 106 & 110 & 124 & 132 \\
                     ~ & RMRT  & 105 & 113 & 115 & 121 & 124  \\
                      \bottomrule
    \end{tabular}
\end{table}   

\textbf{Update time under under varying $\epsilon$.} Table~\ref{tab:epsilon_update_time} further shows the times for inserting 100\% more points (following the distribution of amzn, same below). We see that the insertion times increase with $\epsilon$. This is because the bound on the number of insertions before model rebuilding is inversely proportional to $\epsilon$ (Lemma~\ref{lemma:insert_bound}), i.e., a larger $\epsilon$ triggers rebuilds more eagerly.

\textbf{Update time under varying insertion ratios.}
Next, we test the impact of the number of points inserted. RMI, RMI-NN, and RS are static indices and are omitted for this experiment. 
For the fanout (maximum error in PGM), we use $2^{10}$ for PGM and $2^{13}$ for RMRT and RMI-NN-MR, respectively.
As shown in Fig.~\ref{fig:insert_time_ratio}, the insertion times of BTree, RMI-NN-MR, and RMRT increase with the insertion ratio. For PGM, the insertion time rises in stages. It has higher insertion costs than RMI-NN-MR in most cases (from 10\% to 90\%). This is because PGM uses the logarithmic method~\cite{design_dynamic_data_structures} which builds and merges (hence the cost jumps) a set of PGMs for insertions. 
For RMI-NN-MR, when the insertion ratio is within 80\%, the cost increases slowly because most points are inserted directly. The costs increase faster when the insertion ratio exceeds 80\%, which exceeds the insertion bound and leads to more rebuilds.
For RMRT, the insertion time is more stable. This is because each RMRT sub-model indexes a relatively small set of data points and has a relatively small rebuild bound. Model rebuild is triggered steadily as new data points are inserted. 

\vspace{-0.1in}
\begin{figure}[h]
    \centering
       \subfloat[Varying insertion ratio~\label{fig:insert_time_ratio}]{
    	\includegraphics[width=0.255\textwidth]{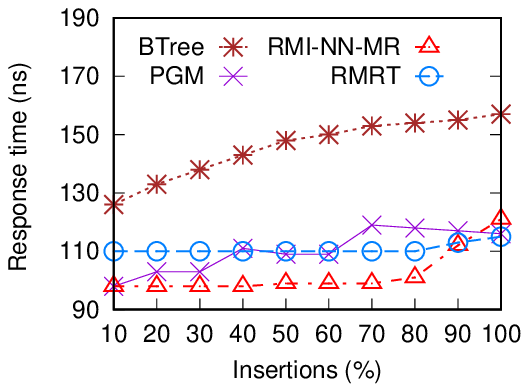}
    	\hspace{-6mm}
    }
     \subfloat[Varying branching parameter~\label{fig:insert_time_branch}]{
    	\includegraphics[width=0.255\textwidth]{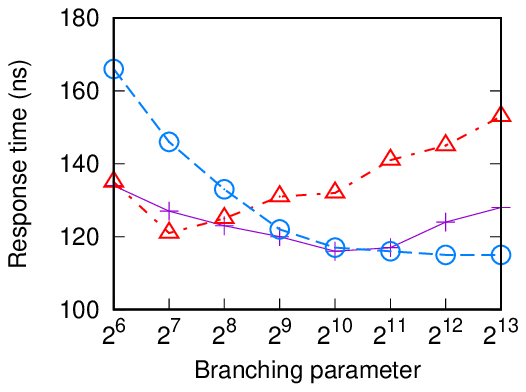}
    }
    \caption{Insertion time results.}\label{fig:insert_time}
\end{figure}

\vspace{2mm}
\textbf{Update time under varying branching parameters.} 
Due to the parameter limitation of dynamic PGM, we vary the fanout (i.e., maximum error in PGM) from $2^{6}$ to $2^{13}$.
As Fig.~\ref{fig:insert_time_branch} shows, 
PGM outperforms RMI-NN-MR when its maximum error (fanout) is larger than $2^8$ because a larger fanout for the RMI learned indices means fewer data points in each model. The cardinality of each model in the second layer is smaller, and the bound for insertion is also smaller (w.r.t.,  Lemma~\ref{lemma:insert_bound}), such that rebuild happens more often. 
For RMRT, the insertion performance becomes better as the fanout increases. This is because it can adaptively divide the underlying models and provide more positions for direct insertions without frequent model rebuilds.



\vspace{1mm}
\vspace{-0.1in}
\subsection{Discussion}
\vspace{-0.5mm}
We note several directions to be explored next. We have omitted the sorting costs in index building, since these are shared by all indices. It would be interesting to further optimize these costs with a learning-based technique. Our CDF similarity approximation considers the maximum distance between the CDFs. An alternative is to take the average distance. How to bound the search range in this case is an interesting challenge.

 \vspace{1mm}
\vspace{-0.1in}
\section{Conclusions}\label{sec:conclusion}
\vspace{-0.5mm}

We proposed to reuse pre-trained models for indexing new (or updated) datasets to address building and update efficiency issues in learned indices. We also proposed a similarity metric to measure the distribution difference between two datasets. Based on this metric, our agile model reuse algorithm can efficiently select the most suitable pre-trained model to index a new (or updated) dataset. We show how the prediction error of the selected pre-trained model is bounded on the new (or updated) dataset. We demonstrate the effectiveness of the proposed algorithm by applying it on the RMI learned indices~\cite{RMI} and our proposed learned index  RMRT.  
Experimental results on synthetic datasets and four real datasets show that our agile model reuse technique can improve the building and update time of learned indices substantially with little impact on the lookup performance. 



\bibliographystyle{ACM-Reference-Format}
\bibliography{references}

\end{document}